\documentclass[pra,twocolumn,showpacs,floatfix,aps,nofootinbib]{revtex4-2}
\usepackage{amsmath,amssymb,epsfig,amsfonts,bbm,graphicx,times,palatino,verbatim,tensor}
\usepackage{subfigure}
\usepackage{color}
\usepackage[english]{babel}
\usepackage{url}
\usepackage{textcase}
\usepackage{bm}
\usepackage[hidelinks]{hyperref}
\usepackage{easybmat}
\usepackage{dsfont}
\usepackage{amsthm}
\usepackage{braket}
\usepackage{ulem}
\usepackage{xcolor}
\usepackage{enumitem}
\usepackage{algorithm}
\usepackage{mathrsfs}

\makeatletter
\newenvironment{protocol}
{
		\renewcommand{\ALG@name}{Protocol}
		\refstepcounter{algorithm}
		\hrule height.8pt depth0pt \kern2pt
		\renewcommand{\caption}[2][\relax]{
			{\raggedright\textbf{\fname@algorithm~\thealgorithm} ##2\par}%
			\ifx\relax##1\relax 
			\addcontentsline{loa}{algorithm}{\protect\numberline{\thealgorithm}##2}%
			\else 
			\addcontentsline{loa}{algorithm}{\protect\numberline{\thealgorithm}##1}%
			\fi
			\kern2pt\hrule\kern2pt
		}
	}{
	\kern2pt\hrule\relax
}
\makeatother

\DeclareMathOperator{\Tr}{Tr}
\newcommand{\id}{\mathbbm{1}}
\def\diracspacing{0.7pt}
\newcommand{\ketbra}[2]{| \hspace{\diracspacing} #1 \rangle \langle #2 \hspace{\diracspacing} |} 
\newcommand{\corr}[1]{\langle #1 \rangle} 
\newcommand{\de}[1]{\left( #1 \right)}

\newcommand{\DE}[1]{\left\{#1\right\}}

\newtheorem{theorem}{Theorem}
\newtheorem{corollary}{Corollary}

\newtheorem{prop}{Proposition}
\newtheorem{propA}{Proposition}
\newtheorem{defin}{Definition}
\numberwithin{propA}{section}

\addto\captionsenglish{}
\addto\captionsenglish{}

\makeatletter
\let\cat@comma@active\@empty
\makeatother

\begin{document}
\normalem

\title{Self-testing with dishonest parties and device-independent entanglement certification in quantum networks}

\author{Gl\'aucia Murta}
\email{glaucia.murta@hhu.de}
\affiliation{Institut f\"ur Theoretische Physik III, Heinrich-Heine-Universit\"at D\"usseldorf, Universit\"atsstra\ss{}e 1, D-40225 D\"usseldorf, Germany}

\author{Flavio Baccari}
\email{flavio.baccari@mpq.mpg.de}
\affiliation{Max-Planck-Institut f\"ur Quantenoptik, Hans-Kopfermann-Stra\ss{}e 1, 85748 Garching, Germany}

\begin{abstract}
Here we consider the task of device-independent certification of the quantum state distributed in a network when
some of the nodes in this network may collude and act dishonestly. We introduce the paradigm of self-testing with dishonest parties and present a protocol to self-test the GHZ state in this framework. We apply this result for state certification in a network with dishonest parties and also provide robust statements about the fidelity of the shared state. Finally, we extend our results to the cluster scenario, where several subgroups of parties may collude during the state certification. Our findings provide a new operational motivation for the strong definition of genuine multipartite nonlocality as originally introduced by Svetlichny in [PRD 35, 3066 (1987)].
\end{abstract}

\maketitle

\section{Introduction}
With recent progress in quantum network communication we are approaching the technological developments required to implement protocols that go beyond point-to-point quantum key distribution (QKD). 
In particular, proof-of-principle implementations of conference key agreement (CKA)~\cite{CKAexperiment}, as well as its generalization to an anonymous setup~\cite{ACKA_Pappa,ACKA}, have been recently demonstrated~\cite{ACKAexp,ACKAexpcluster}, showing that the power of multipartite entanglement can be already explored.

Genuine multipartite correlations, as for example the one encountered in the the Greenberger-Horne-Zeilinger (GHZ) state~\cite{ghz1989}, constitute the essential resource for important network tasks, such as  secret sharing~\cite{secretsharing}, multiparty quantum computation~\cite{n-computation},  and anonymous quantum transmission~\cite{christandl_quantum_2005}.  
Therefore, certifying the entanglement properties of the state distributed in a network, by a potentially untrusted source, is essential to ensure the correct implementation of such network tasks.

Entanglement certification in a network can be achieved with different adversarial levels. If the parties in the network are honest and trust their measurement apparatuses, entanglement can be verified using quantum state tomography or entanglement witness schemes~\cite{Witness}. If, however, the devices of some (steering scenario) or all the parties (device-independent scenario) are untrusted, i.e. they could be partially characterized or potentially produced by an unturstworthy provider, then multipartite entanglement can be certified using steering inequalities~\cite{GMEsteering} or Bell inequalities~\cite{DIEW}. In particular, in the device-independent scenario, self-testing results allows us to make strong statements about the precise form of the shared state~\cite{st-review}.
Finally, in the network scenario, we can have yet another adversarial level, namely some of the parties in the network may be dishonest and, in particular, collude with each other in order to jeopardize the state certification. In protocols where information needs to be concealed from some of the parties in the network, such as in secret sharing and anonymous communication, the parties have an incentive to act maliciously throughout the protocol in order to try to access the hidden information.

The task of entanglement certification in a network with dishonest parties was first considered in~\cite{Pappa12}. Such a verification scheme lifted the anonymous communication protocol of \cite{christandl_quantum_2005} to the untrusted source scenario~\cite{AT_verified}. Subsequently the certification protocol of~\cite{Pappa12} was improved and implemented in~\cite{McCutcheon16}, and more recently generalized to certify all graph states ~\cite{Unnikrishnan22}.

Here we consider the task of entanglement certification in a quantum network with dishonest parties and uncharacterized devices.
We introduce
the paradigm of self-testing with dishonest parties and present a protocol to self-test the GHZ state
in this framework. We apply this result to design a protocol to certify the GHZ state in a network with
dishonest parties and also provide robust statements about the fidelity of the shared state.
Finally, we extend our results to the cluster scenario, where different subgroups of parties may collude during the state certification.

\section{The network scenario}

We consider a network with $N$ parties/nodes and a source that distributes an $N$-partite state.  Every pair of parties in the network is connected by a private classical channel. The parties in the network may be honest or dishonest.  Let $\mathcal{H}$, $|\mathcal{H}|=k-1$, represent the set of honest parties, and  $\mathcal{D}$, $|\mathcal{D}|=N-k+1$, be the set of dishonest parties. While parties in $\mathcal{H}$ are assumed to follow all the specifications of the protocol, the unknown subset of $N-k+1$  dishonest parties may deviate arbitrarily from the protocol's description and even control the source in order to jeopardize the state certification.
Additionally, we do not make any assumptions about the internal working of the devices of the honest parties, i.e, we consider a device-independent scenario.
The goal of the parties is to certify, in a device-independent way, the state distributed by the source in the presence of potential dishonest parties.
For this work we assume an IID (identically and independently distribute) setup, i.e., that the distributed quantum state    and the strategies applied by the honest/dishonest parties are the same in every round.

For the proposed certification scheme, we  will consider a Bell scenario involving dishonest parties. Assume that each party receives one dichotomic input {$x_i \in \DE{0,1}$} and has to provide a dichotomic output {$a_i \in \DE{0,1}$}. Upon collecting many rounds of outputs, the corresponding statistics is described by a collection of conditional distributions $p(a_1 \ldots a_N|x_1 \ldots x_N)$. The most general classical correlations achieved in the network scenario are then described by
\begin{align}\label{eq:netlocalmodel}
\begin{split}
 p(a_1a_2\ldots a_N&|x_1x_2\ldots x_N)\\
 &=   \int d\rho_{\lambda} p(\vec{a}_{\mathcal{D}} | \vec{x}_{\mathcal{D}}, \lambda) \prod_{i \in \mathcal{H}} p(a_i | x_i, \lambda) ,
 \end{split}
\end{align}
where $\vec{a}_{\mathcal{D}}$ and $\vec{x}_{\mathcal{D}}$ are vectors collecting the outputs and inputs of all the dishonest parties. Note that $p(\vec{a}_{\mathcal{D}} | \vec{x}_{\mathcal{D}})$ is allowed to be an arbitrary (even signaling) probability distribution, to account for the most general strategy that the dishonest parties can apply. 

In a quantum realization, the action of the honest parties is described by measuring a binary observable $A_{x_i}^{(i)} = \Pi_{0| x_i} - \Pi_{1| x_i}$, where $\Pi_{a_i| x_i}$ is the POVM element associated with outcome $a_i$ for party $i$. In contrast, we associate a global observable for the action of the dishonest parties. Since we will consider a Bell inequality that only depends on the parity of the dishonest parties' outcomes, $J(\vec{a}_{\mathcal{D}}) =\bigoplus_{j\in\mathcal{D}}a_j$, we can define the following observable to describe their action
\begin{equation}\label{eq:dishonest_corr}
M_{\vec{x}_{\mathcal{D}}}^{(\mathcal{D})} = \sum_{ J(\vec{a}_{\mathcal{D}}) = 0 } \Pi_{\vec{a}_{\mathcal{D}} | \vec{x}_{\mathcal{D}}}  - \sum_{ J(\vec{a}_{\mathcal{D}}) = 1 } \Pi_{\vec{a}_{\mathcal{D}} | \vec{x}_{\mathcal{D}}}  \, ,
\end{equation}
where $\Pi_{\vec{a}_{\mathcal{D}} | \vec{x}_{\mathcal{D}}}$ is the POVM element associated with the string $\vec{a}_{\mathcal{D}}$ of outcomes of the dishonest parties for input $\vec{x}_{\mathcal{D}}$.

A key ingredient of our result is a Bell inequality that witnesses genuine multipartite nonlocality in the sense that was first introduced by Svetlichny~\cite{Svetlichny3}. Specifically we make use of the family of $N$-partite Svetlichny inequalities~\cite{Svetlichny3,SvetlichnyN} defined by the expression:
\begin{align}\label{eq:defS}
    S_N^{\pm}= \sum_{\vec{x}}(-1)^{\frac{w_{\vec{x}}(w_{\vec{x}}\pm 1)}{2}}\corr{A_{x_1}^{(1)}A_{x_2}^{(2)}\ldots A_{x_N}^{(N)}},
\end{align}
where $\vec{x}=(x_1, x_2,\ldots,x_N)\in \DE{0,1}^{\times N}$ is a string of $N$ bits that labels the parties' inputs, and $w_{\vec{x}}$ is the Hamming weight of string $\vec{x}$. Moreover the correlators are defined as
\begin{align}
\corr{A_{x_1}^{(1)}A_{x_2}^{(2)}\ldots A_{x_N}^{(N)}}=\sum_{J(\vec{a}) = 0}p(\vec{a}| \vec{x}) - \sum_{J(\vec{a}) = 1}p(\vec{a}| \vec{x})
\end{align}
where $\vec{a}=(a_1, a_2,\ldots,a_N)\in \DE{0,1}^{\times N}$ is the string of outcomes, and $J(\vec{a}) =\bigoplus_{j=1}^N a_j$ is the partity of string $\vec{a}$.

The Svetlichny inequalities read
\begin{align} \label{eq:SN}
 |S_N^{\pm}|\stackrel{\mathcal{L}}{\leq} 2^{N-1}\stackrel{\mathcal{Q}}{\leq}2^{N-1}\sqrt{2},
\end{align}
where $\mathcal{L}/\mathcal{Q}$ denotes the classical/quantum bound. The classical bound,  $ |S_N^{\pm}|\leq 2^{N-1}$, constraints all the distributions that can be decomposed into the form
\begin{align}
\label{eq:localmodel}
\begin{split}
 p(a_1a_2&\ldots a_N|x_1x_2\ldots x_N)\\
 &=   \int d\rho_{\lambda}\sum_{\mathcal{P}\subsetneq\DE{1,\ldots, N}}p(\vec{a}_{\mathcal{P}}|\vec{x}_{\mathcal{P}}\lambda)p(\vec{a}_{{\mathcal{P}^c}}|\vec{x}_{{\mathcal{P}^c}}\lambda),
 \end{split}
\end{align}
where $p(\vec{a}_{\mathcal{P}}|\vec{x}_{\mathcal{P}}\lambda)$ is an arbitrary distribution between the parties in set subset $\mathcal{P}$, and $\mathcal{P}^c$ denotes the complementary set.

By grouping the dishonest parties together, with respective observables $M_{\vec{x}_{\mathcal{D}}}^{(\mathcal{D})}$, as defined in \eqref{eq:dishonest_corr}, the Svetlichny expression \eqref{eq:defS} can  be written in terms of $k$-partite correlators, involving the $k-1$ observables of honest parties and a joint observable of the dishonest group. The symmetries of the Svetlichny inequalities ensure that a violation of the $N$-partite Svetlichny inequality in this setting implies a violation of a $k$-partite Svetlichny inequality for the non-communicating honest parties and the group of dishonest parties, as stated in the following proposition. The proof of Proposition~\ref{prop:sk_violation} is presented in Appendix~\ref{app:prop-proofs}. 

\begin{prop}\label{prop:sk_violation}
If a strategy achieves value $s_N$ for the $N$-partite Svetlichny inequality, then  the same strategy  achieves value $s_{k}$ for a $k$-partite Svetlichny inequality, with
\begin{align}
    s_{k}\geq \frac{s_N}{2^{N-k}},
\end{align}
where  $k-1$ parties perform their respective individual strategy, and $N-k+1$ parties are grouped together, potentially performing a joint strategy, with their joint outcome defined by $a'=J(\vec{a}_{\mathcal{D}})$.
\end{prop}

\section{Self-testing with dishonest parties}

Proposition~\ref{prop:sk_violation} shows that a violation of the $N$-partite Svetlichny inequality witnesses the violation of a $k$-partite Svetlichny inequality by the honest parties and the set of dishonest parties.
Now we prove that a strong characterization of the distributed state and measurements can be achieved when the maximal violation of the Svetlichny inequality is observed.
For that, let us introduce a notion of self-testing that is suitable to the scenario with dishonest parties. To make the following expressions more concise, we take $\mathcal{D}=\DE{k,k+1\ldots,N}$, which can always be obtained by relabeling the parties.
\begin{defin}\label{definition}
A set of observed correlations $p(\vec{a}|\vec{x})$ self-tests the $k$-partite state $\ket{\Phi}$ in a dishonest parties scenario if, for any state $\rho$ with purification $\ket{\psi}$ compatible with $p(\vec{a}|\vec{x})$ for some measurements described by observables $A_{x_{1}}^{(1)}, \ldots, A_{x_{k-1}}^{(k-1)},  M_{\vec{x}_\mathcal{D}}^{(\mathcal{D})}$, there exist local isometries  $\lbrace \Lambda_i \rbrace_{i = 1}^{k-1}$ for the honest parties and a global isometry $ \Lambda_{\mathcal{D}}$ for the dishonest parties such that
\begin{align}\label{stsub}
   \Lambda_1 \otimes \ldots \otimes
   \Lambda_{k-1} \otimes \Lambda_{\mathcal{D}} (\ket{\psi}) = \ket{\Phi} \otimes \ket{\zeta}
\end{align}
where $\ket{\zeta}$ denotes some uncorrelated degrees of freedom. Additionally, the same correlation also self-tests a set of target measurements ${\bar{A}^{(i)}_0, \bar{A}^{(i)}_1}$, for $i = 1,\ldots,k-1$, and $\bar{A}_{\vec{x}_\mathcal{D}}^{(k)}$, for $\vec{x}_\mathcal{D} \in \DE{0,1}^{\times |\mathcal{D}|}$,   if it follows that, for all input choices,
\begin{align}\label{opsub}
   \Lambda_1 \otimes \ldots &\otimes
   \Lambda_{k-1} \otimes \Lambda_{\mathcal{D}} ( A_{x_1}^{(1)} \otimes \ldots \otimes
   A_{x_{k-1}}^{(k-1)} \otimes M_{\vec{x}_\mathcal{D}}^{(\mathcal{D})} \ket{\psi}) \nonumber \\
   &= ( \bar{A}_{x_1}^{(1)} \otimes \ldots \otimes
   \bar{A}_{x_{k-1}}^{(k-1)} \otimes \bar{A}_{\vec{x}_\mathcal{D}}^{(k)} \ket{\Phi})\otimes \ket{\zeta}.
\end{align}
\end{defin}
Note that the self-testing statement accounts for collective operations involving all the dishonest parties, since one cannot a priori exclude that all the $N-k+1$ parties are acting maliciously and cooperating as a joint effective party. For that reason, the best one can hope for is to self-test an entangled state shared between the honest parties and a single additional party, representing the dishonest ones as a collective. In other words, the self-tested state  $\ket{\psi}$ belongs to a $k$-partite Hilbert space.
With that in mind, we are able to show what follows.
\begin{theorem}\label{thm:self-test}
The maximum violation of the $N$-partite Svetlichny inequality with a set of dishonest parties $\mathcal{D}$, $|\mathcal{D}|=N-k+1$, self-tests that a $k$-partite GHZ state is shared by the honest parties and the set of dishonest parties. Moreover, the same correlations also self-test a set of Pauli observables for the measurements performed by the honest parties and the joint measurements of the dishonest parties. 
\end{theorem}
The proof of Thm.~\ref{thm:self-test} and a detailed description of the self-tested measurements are presented in Appendix~\ref{app:net-selftest}. 

Thm.~\ref{thm:self-test} is a stronger form of self-testing statement that allows us to infer the existence of a specific shared state even in the presence of dishonest parties.
Indeed, when the dishonest parties implement a joint measurement that leads to the maximal violation of the $N$-partite Svetlichny inequality, the resulting statistics will maximally violate a $k$-partite Svetlichny inequality, where now we have the standard Bell scenario with the $k-1$ honest parties and the group of dishonest parties performing local measurements.
As a collorary of Thm.~\ref{thm:self-test}, when $|\mathcal{D}|\leq 1$ we obtain a standard self-testing result for the  Svetlichny inequalities. 
\begin{corollary}\label{prop:self-test}
The maximal violation of an $N$-partite Svetilichny inequality, in the standard Bell scenario, self-tests the $N$-partite GHZ state, and the respective Pauli observables that lead to maximal violation of $S_N^{\pm}$.
\end{corollary}
 Note that even though the family of MABK Bell inequalities~\cite{MABK1,MABK2,MABK3} can be used to self-test the GHZ state in the standard Bell scenario, it fails to provide a self-testing statement in the presence of dishonest parties.
Indeed, the strong form of genuine multipartite nonlocality witnessed by the Svetlichny inequality seems to be a crucial ingredient for self-testing in the presence of dishonest parties. 
It is worth mentioning that different definitions of genuine multipartite nonlocality have been introduced~\cite{Bancal13,Gallego12}, where the decomposition in \eqref{eq:localmodel} is restricted to non-signalling or time-ordered distributions with one-way signalling. In particular, in~\cite{Gallego12} it is shown that Svetlichny's original definition of genuine multipartite nonlocality is inconsistent with a general operational framework for nonlocality.
Nevertheless, Svetlichny's strong definition of genuine multipartite nonlocality is appropriate in our scenario because we consider a setup where the dishonest parties may collude and perform a joint strategy, which is fairly captured by  a  signalling probability distribution. The strong form of nonlocality witnessed by the Svetlichny inequality was also shown to have potential application for device-independent secret sharing~\cite{Moreno2020,secretsharingNQKD}.

\section{Genuine multipartite entanglement certification in quantum networks}

Using the previous results, we
 now introduce a protocol for device-independent entanglement certification in a network with dishonest parties. 
\vspace{0.5em}

\begin{protocol}\label{prot}
	\caption{}
		One of the parties, denoted the Verifier, will proceed to certify the state generated by the source. W.l.o.g. we can assume the Verifier to be party $A_1$.
		\begin{enumerate}[label=\arabic*., ref=\arabic*]
		\item Repeat several times:
		\begin{enumerate}[label*=\arabic*.]
		\item\label{prot:step1} For each $i\in\DE{1,\ldots,N}$, the Verifier selects a random input $x_i\in\DE{0,1}$. The Verifier keeps their corresponding input $x_1$ and sends $x_i$ to party $A_i$ using a private channel.			\item\label{prot:step2} Upon receiving input $x_i$, party $A_i$ produces  output $a_i$ and sends it  to the Verifier using  a private channel.
	\end{enumerate}
			\item\label{prot:step4} The Verifier computes the value $s_N$ for the Svetlichny inequality, $S_N^+$, from the observed distribution of inputs and outputs. 
		\end{enumerate}
	\end{protocol}\vspace{1em}

Our first result is a qualitative statement about the entanglement properties of the distributed state.

\begin{theorem}\label{thm:gme} If an honest Verifier observes a violation of the Svetlichny ienquality,  $s_N>2^{N-1}$, then Protocol~\ref{prot} certifies genuine multipartite entanglement between the  honest parties and the set of unknown dishonest parties.
\end{theorem}

\begin{proof}[Proof of Theorem~\ref{thm:gme}] By  Proposition~\ref{prop:sk_violation}, the violation of the $N$-partite Svetlichny inequality implies that the $k-1$ honest parties and the set of dishonest parties also violate a $k$-partite Svetlichny inequality in the standard $k$-partite Bell scenario. Moreover, the definition of the local bound for the Svetlichny inequality, eq.~\eqref{eq:localmodel}, implies that a violation of a $k$-partite inequality in the standard Bell scenario witnesses $k$-partite entanglement.
\end{proof}

The self-testing results provided in the previous section allow us to go beyond a qualitative detection of genuine multipartite entanglement, and instead we can make a statement about the shape of the shared state. We recall that in the scenario with dishonest parties the best one can do is to certify the state up to a joint operation on the dishonest parties' systems (see~\cite{Pappa12}).

\begin{theorem}\label{thm:ghz}
If the Verifier is honest and the maximal violation of the Svetlichny inequality is observed. Protocol~\ref{prot} certifies the $N$-partite GHZ state up to local isometries on the honest parties and a global isometry on the dishonest parties. 
\end{theorem}

\begin{proof}
By Theorem~\ref{thm:self-test} and the definition of self-testing with dishonest parties, the maximal violation of the Svetlichny inequality implies that for any state $\rho$ with purification $\ket{\psi}$ shared by the parties, it holds that
\begin{align}
    \Lambda_1 \otimes \ldots \otimes
   \Lambda_{k-1} \otimes \Lambda_{\mathcal{D}} (\ket{\psi}) = \ket{\Phi^k} \otimes \ket{\zeta}
\end{align}
for some isometries $ \Lambda_1  \ldots,
   \Lambda_{k-1},\Lambda_{\mathcal{D}}$, where  $\ket{\Phi^k}=\frac{1}{\sqrt{2}}(\ket{0 \ldots 0}+\ket{1\ldots 1})$ is the $k$-partite GHZ state.
Now, using the freedom of operations on the dishonest set, an additional CNOT operation on the set of dishonest parties brings the state to the target $N$-partite state
\begin{align}
\begin{split}
    \Lambda_1 \otimes & \ldots \otimes
   \Lambda_{k-1} \otimes \Lambda'_{\mathcal{D}} (\ket{\psi}) \\
   &= \Lambda_1 \otimes \ldots \otimes
   \Lambda_{k-1} \otimes {\rm CNOT}^{N-k}_{\mathcal{D}}\circ \Lambda_{\mathcal{D}} (\ket{\psi}) \\
   &= {\rm CNOT}^{N-k}_{\mathcal{D}}(\ket{\Phi^k} \otimes \ket{\zeta})=\ket{\Phi^N} \otimes \ket{\zeta}
   \end{split}
\end{align}
where ${\rm CNOT}^{N-k}_{\mathcal{D}}$ corresponds to the CNOT gate applied to $N-k$ ancillas on the state $\ket{0}$ with control on system $\mathcal{D}$, and $\ket{\Phi^N}$ is the GHZ state of $N$ parties.
\end{proof}

Thm.~\ref{thm:ghz} resembles the certification guarantees as first defined in~\cite{Pappa12}. The difference here is that, in the device-independent scenario, the  state is certified up to local isometries in the honest parties.

The strong self-testing properties of the Svetlichny inequalities can also provide robust guarantees about the distributed state.
We will now derive bounds on the device-independent fidelity of the shared state as a function of the observed violation.
For that,  we define the following figure of merit for the network scenario
\begin{align}\label{eq:F}
    F_{\rm DI}^{\mathcal{D}}(s_N)=\inf_{\tilde{\rho}\in \mathcal{S}({s}_N^\mathcal{D})} \max_{\substack{\Lambda_{\mathcal{D}},\Lambda_{P_i}\\P_i\in\mathcal{H}}}F(\otimes_{P_i\in \mathcal{H}}\Lambda_{P_i}\otimes\Lambda_{\mathcal{D}}(\tilde{\rho}),\Phi^N),
\end{align}
where $\Phi^N=\ketbra{\Phi^N}{\Phi^N}$, $\Lambda_{P_i}$ denote local channels on the system of the honest parties, $P_i \in \mathcal{H}$, and $\Lambda_{\mathcal{D}}$ is a joint quantum channel on the systems of the dishonest parties $\mathcal{D}$, and $\mathcal{S}({s}_N^\mathcal{D})$ is the set of all quantum states that achieves value at least $s_N$ for an $N$-partite Svetlichny inequality when the parties in $\mathcal{D}$ can apply a joint strategy. The fidelity is defined as $
    F(\rho, \sigma)=\de{\Tr |\sqrt{\rho}\sqrt{\sigma} | }^2$.
    
Eq.~\eqref{eq:F} generalizes the concept of \emph{extractability} introduced in~\cite{JedSelftest} to the dishonest parties' scenario. 
Bounds for the fidelity as a function of the Bell violation in the standard Bell scenario can be derived using the self-testing
from operator inequalities (STOPI) method introduced in~\cite{JedSelftest} (see also~\cite{Tim-st}). The STOPI method  generates bounds of the following type for a $k$-partite Svetlichny inequality
\begin{align}\label{F_sk}
    F_{DI}(s_k)\geq f_k s_k - \mu_k
\end{align}
where $F_{DI}(s_k)$ correspond to eq.~\eqref{eq:F} for $|\mathcal{D}|\leq 1$, and $f_k$ and $\mu_k$ are constants that depend on the Svetlichny inquality. Note that  we drop the superscript $\mathcal{D}$ to highlight that eq.~\eqref{F_sk} refers to the standard Bell scenario with $k$ non-collaborating parties.

In the following theorem, we show that bounds for the standard Bell scenario, \eqref{F_sk}, can be used to bound our quantity of interest,  $F_{\rm DI}^{\mathcal{D}}(s_N)$.
\begin{theorem}\label{thm:F_D_k}
If an honest Verifier observes a violation $s_N$ in Protocol~\ref{prot}, then the following fidelity can be certified
\begin{align}
    F_{DI}^{|\mathcal{D}|=N-k+1}\geq f_k \frac{s_N}{2^{N-k}}-\mu_k
\end{align}
where $f_k$ and $\mu_k$ are the constants for the $k$-partite Svetlichny inequalities, as defined in \eqref{F_sk}.
\end{theorem}

\begin{proof}
The proof follows from the  chain of inequalities:
\begin{align}
F_{DI}^{|\mathcal{D}|=N-k+1}&\nonumber\\
=\inf_{\tilde{\rho}\in \mathcal{S}({s}_N^\mathcal{D})} &\max_{\substack{\Lambda_{\mathcal{D}},\Lambda_{P_i}\\P_i\in\mathcal{H}}}F(\otimes_{P_i\in \mathcal{H}}\Lambda_{P_i}\otimes\Lambda_{\mathcal{D}}(\tilde{\rho}),\Phi^N)\\
\geq \inf_{{\rho}\in \mathcal{S}({s}'_k)} &\max_{\substack{\Lambda_{\mathcal{D}},\Lambda_{P_i}\\P_i\in\mathcal{H}}}F(\otimes_{P_i\in \mathcal{H}}\Lambda_{P_i}\otimes\Lambda_{\mathcal{D}}({\rho}),\Phi^N)\label{step1}\\
\geq  \inf_{{\rho}\in \mathcal{S}({s}'_k)} &F\left(\otimes_{P_i\in \mathcal{H}}\Lambda^*_{P_i}\otimes {\rm CNOT}_{\mathcal{D}}^{N-k} \circ \Lambda^*_{\mathcal{D}}({\rho}),\right.\label{step2}\\
&\left. \;\;\;\;\;\;\;\;\; {\rm CNOT}_{\mathcal{D}}^{N-k}(\Phi^k)\right)\nonumber\\
\geq \inf_{{\rho}\in \mathcal{S}({s}'_k)}& F(\otimes_{P_i\in \mathcal{H}}\Lambda^*_{P_i}\otimes  \Lambda^*_{\mathcal{D}}({\rho}),\Phi^k)\label{step3}\\
= \inf_{{\rho}\in \mathcal{S}({s}'_k)}&\max_{\substack{\Lambda_{\mathcal{D}},\Lambda_{P_i}\\P_i\in\mathcal{H}}} F(\otimes_{P_i\in \mathcal{H}}\Lambda^*_{P_i}\otimes  \Lambda^*_{\mathcal{D}}({\rho}),\Phi^k)\label{step4}\\
 \geq f_k  \,& \frac{s_N}{2^{N-k}} -\mu_k \label{step5}
  \end{align}
 where $s'_k=\frac{s_N}{2^{N-k}}$. In~\eqref{step1} we use the fact that every state that violates an $N$-partite Svetlichny inequality with dishonest set $\mathcal{D}$ by an amount $s_N$, also violates  a $k$-partite inequality by $s_k=s'_k$. For~\eqref{step2}, we fix a particular map, where $\Lambda^*_{P_i}, \Lambda^*_{\mathcal{D}}$ are chosen to  be the maps that maximize \eqref{step4} and ${\rm CNOT}^{N-k}$ corresponds to the CNOT gate applied to $N-k$ ancillas on the state $\ket{0}$ with control on system $\mathcal{D}$. Inequality~\eqref{step3} follows from data-processing
 $F(\Lambda(\rho),\Lambda(\sigma))\geq F(\rho,\sigma)$. Finally \eqref{step5} corresponds to the bound for the standard $k$-partite Bell scenario.
\end{proof}

Note that the bounds $F_{DI}^{|\mathcal{D}|=N-k+1}$ for different values of $k$ may not be ordered, because as $k$ increases the infimum is taken over a smaller set of states (i.e. those achieving the given violation with fewer dishonest parties), but the maximum is also taken over a smaller set of operations (since we have more honest parties constrained to acting locally). 
This is illustrated in Fig.~\ref{fig:st-numeric}, where we plot bounds for $F_{DI}^{|\mathcal{D}|=N-k+1}$ based on numerical evidence obtained using the STOPI technique~\cite{JedSelftest,Tim-st}. Details of the numerical evidence are presented in Appendix~\ref{app:robust}.

\begin{figure}[H]
    \centering
    \includegraphics[width=0.42\textwidth]{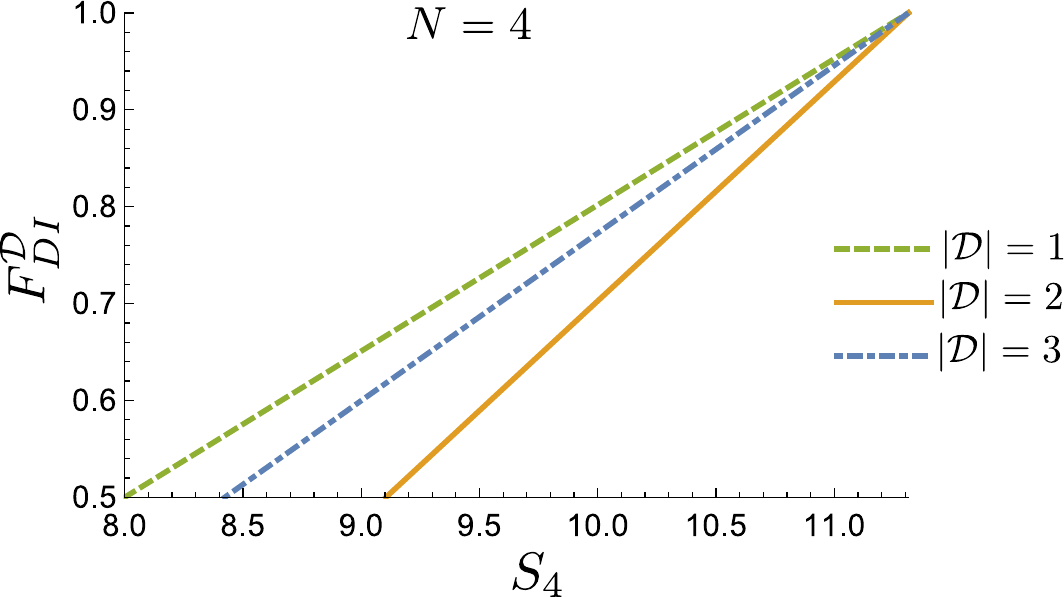}
    \caption{Numerical lower bounds on the fidelity as function of the violation of a 4-partite Svetlichny inequality. The curves represent scenarios with different number of dishonest parties.}
    \label{fig:st-numeric}
\end{figure}

Thm.~\ref{thm:F_D_k} gives a route to bound $ F_{\rm DI}^{\mathcal{D}}$ for an arbitrary set of dishonest parties of unknown size, as we can consider the worst case scenario 
\begin{align}
     F_{\rm DI}^{\mathcal{D}}\geq \min_k \DE{F_{DI}^{|\mathcal{D}|=N-k+1}}.
\end{align}
This provides a robust device-independent state certification in a network with dishonest parties.

\section{Self-testing with cooperating clusters}

Finally, we remark that the results derived in the previous sections can be extended to the scenario where different subgroups of parties may collude. We denote it a network with cooperating clusters, see Fig.~\ref{fig:clusters}.  
This scenario is motivated by a 
network where specific sets of parties are more likely to collaborate with each other.

\begin{figure}[H]
    \centering
    \includegraphics[width=0.49\textwidth]{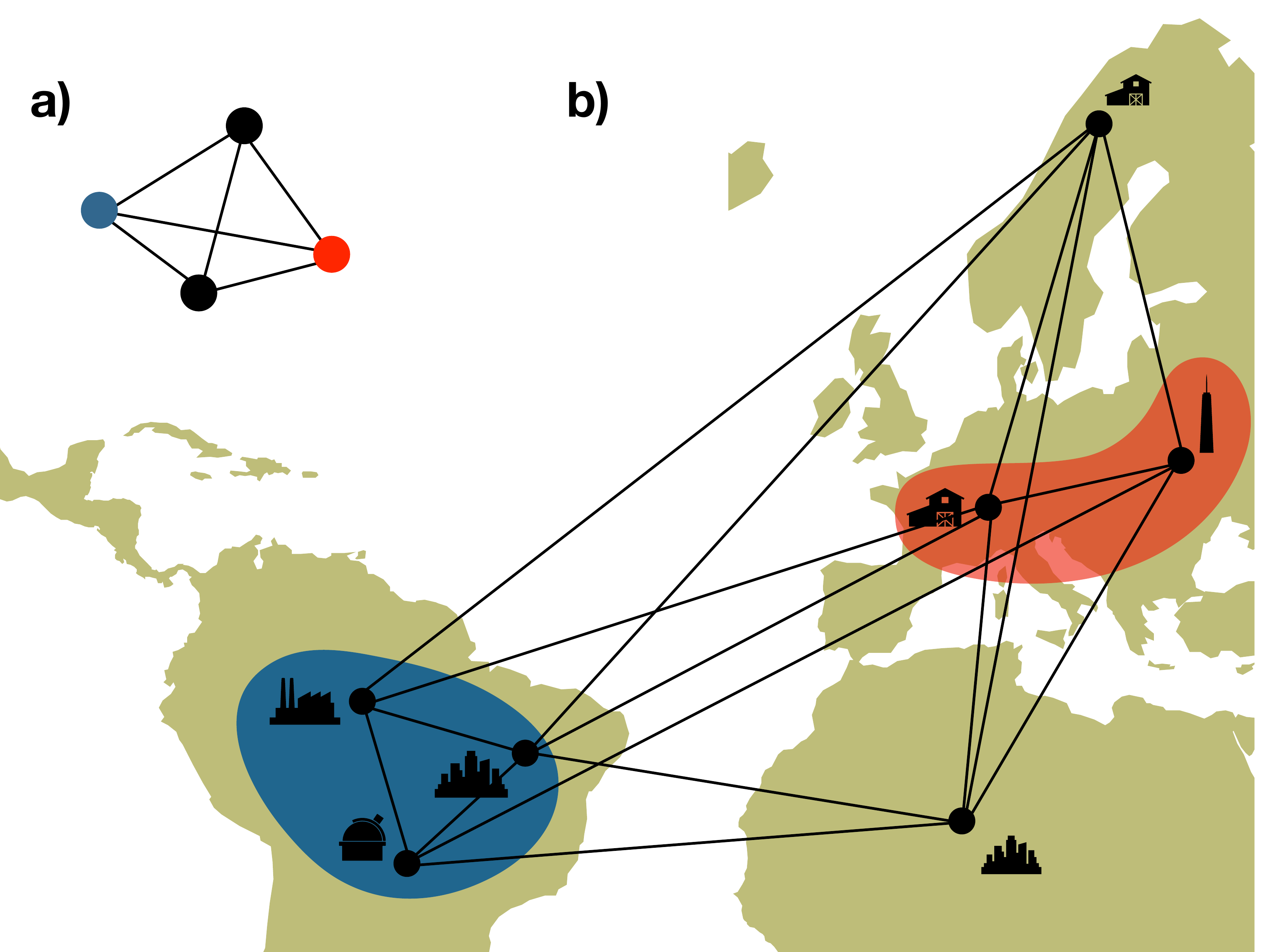}
    \caption{a) Effective network for self-testing, where each cluster is treated as a single party. b) Corresponding pictorial representation of a network with cooperating clusters. Shaded blue and red regions indicate nodes that are likely to collude.} 
    \label{fig:clusters}
\end{figure}

The results also follow from the symmetries of the Svetlichny inequalities. Indeed, in a scenario with $k$ disjoint clusters, a violation of the $N$-partite Svetlichny inequality also implies the violation of a $k$-partite Svetlichny inequality where each party  represents one cluster. Therefore it is straightforward to see that all the self-testing and entanglement certification results derived in the previous section also extend to the cluster scenario.  For more details see Appendix~\ref{app:clusters}.

\section{Discussion}
We investigated the task of device-independent state certification in a network with dishonest parties.
We introduced the concept of self-testing with dishonest parties, and proved self-testing of the GHZ state and Pauli measurements based on the $N$-partite Svetlichny inequalities.

The Svetlichny inequalities, with their ability to witness strong multipartite nonlocality, as defined in~\eqref{eq:localmodel}, seems to be a crucial ingredient for our results. 
Indeed, we conjecture that witnessing genuine multipartite nonlocality in the sense originally defined by Svetlichny is necessary for self-testing and device-independent entanglement certification with dishonest parties.

We applied the self-testing results to design a protocol that can certify the GHZ state in a network with dishonest parties. Our protocol can  witness genuine multipartite entanglement among the honest parties and the set of dishonest parties and also allows us to make strong statements about the fidelity of the distributed state with an $N$-partite GHZ state in the dishonest scenario.
Our results are proved under the IID assumption, i.e. that the state distributed by the source and the strategies of the parties are the same in every round of the protocol. An interesting outlook is to investigate how to drop this assumption in order to consider a fully adversarial scenario.

The results derived in this work have direct application to relax the trusted source assumption for network protocols based on the GHZ state, similarly to what was done in~\cite{AT_verified}. It is interesting to ask whether we can  extend the results to  self-test other classes of multipartite state in the presence of dishonest parties. In particular a protocol to certify the W-state~\cite{Wstate} can lift the anonymous transmission protocol of~\cite{AT_Wstate} to the untrusted source scenario.

\section*{Acknowledgement}
We thank Hermann Kampermann for helpful comments, Tim Coopmans for sharing his master thesis with details on the results of~\cite{Tim-st}, Boris Bourdoncle for giving us the inspiration to improve Fig. 2 and Ivan {\v{S}}upi{\'{c}} for feedback on an earlier version of this manuscript.   GM is funded by the Deutsche Forschungsgemeinschaft (DFG, German Research Foundation) under Germany's
Excellence Strategy - Cluster of Excellence Matter and
Light for Quantum Computing (ML4Q) EXC 2004/1 -
390534769. FB acknowledges financial support from the Alexander von Humboldt Foundation.\\
$\ddag$ The authors contributed equally to this work.

\bibliography{biblio}

\widetext

\appendix

\section{Proof of Proposition 1}\label{app:prop-proofs}

\begin{propA}\label{prop:appsk_violation}
If a strategy achieves value $s_N$ for the $N$-partite Svetlichny inequality, then  the same strategy  achieves value $s_{k}$ for a $k$-partite Svetlichny inequality, with
\begin{align}
    s_{k}\geq \frac{s_N}{2^{N-k}},
\end{align}
where  $k-1$ parties perform their respective individual strategy, and $N-k+1$ parties are grouped together with their joint outcome defined by $a'=J(\vec{a}_{\mathcal{D}})$.
\end{propA}

\begin{proof}
Without loss of generality, let us take $\mathcal{D}=\DE{k,k+1\ldots,N}$, and consider the $N$-partite Svetlichny inequality $S_N^+$ given in Eq.~\eqref{eq:SN}. The action of the dishonest parties is now represented by the joint binary observables $M^{(\mathcal{D})}_{x_kx_{k+1}\ldots x_{N}}$ with outcome $a'=\bigoplus_{i=k}^Na_i$, i.e.
\begin{align}\label{eq:untrustedsubs}
    A_{x_{k}}^{(k)}A_{x_{k+1}}^{(k+1)} \ldots     A_{x_{N}}^{(N)}\rightarrow M^{(\mathcal{D})}_{x_kx_{k+1}\ldots x_{N}}.
\end{align}
By substituting $M^{(\mathcal{D})}_{x_kx_{k+1}\ldots x_{N}}$ in the Svetlichny inequality, we get
\begin{align} \label{eq:SN+_dec}
\begin{split}
  S_N^+  =&\sum_{\vec{x}_{N-k}\in\DE{0,1}^{\times N-k}}\sum_{\vec{x}_{k}\in\DE{0,1}^{\times k}}(-1)^{\frac{(\omega_{\vec{x}_k}+\omega_{\vec{x}_{N-k}})\cdot(\omega_{\vec{x}_k}+\omega_{\vec{x}_{N-k}}+1)}{2}}\corr{A_{x_1}^{(1)}\ldots A_{x_{k-1}}^{(k-1)}M^{(\mathcal{D})}_{x_{k}x_{k+1}\ldots x_{N}}}\\
     =&\sum_{\vec{x}_{N-k}\in\DE{0,1}^{\times N-k}} (-1)^{\frac{\omega_{\vec{x}_{N-k}}\cdot(\omega_{\vec{x}_{N-k}}+1)}{2}}\sum_{\vec{x}_{k}\in\DE{0,1}^{\times k}}(-1)^{\frac{\omega_{\vec{x}_k}\cdot(\omega_{\vec{x}_k}+1)}{2}+\omega_{\vec{x}_{k}}\omega_{\vec{x}_{N-k}}}\corr{A_{x_1}^{(1)}\ldots A_{x_{k-1}}^{(k-1)}M^{(\mathcal{D})}_{x_{k}x_{k+1}\ldots x_{N}}}\\
      =&\sum_{\substack{\vec{x}_{N-k}\in\DE{0,1}^{\times N-k}\\\omega_{\vec{x}_{N-k}} \text{ odd}}} (-1)^{\frac{\omega_{\vec{x}_{N-k}}\cdot(\omega_{\vec{x}_{N-k}}+1)}{2}}\sum_{\vec{x}_{k}\in\DE{0,1}^{\times k}}(-1)^{\frac{\omega_{\vec{x}_k}\cdot(\omega_{\vec{x}_k}-1)}{2}}\corr{A_{x_1}^{(1)}\ldots A_{x_{k-1}}^{(k-1)}M^{(\mathcal{D})}_{x_{k}x_{k+1}\ldots x_{N}}}\\
      &+\sum_{\substack{\vec{x}_{N-k}\in\DE{0,1}^{\times N-k}\\\omega_{\vec{x}_{N-k}} \text{ even}}} (-1)^{\frac{\omega_{\vec{x}_{N-k}}\cdot(\omega_{\vec{x}_{N-k}}+1)}{2}}\sum_{\vec{x}_{k}\in\DE{0,1}^{\times k}}(-1)^{\frac{\omega_{\vec{x}_k}\cdot(\omega_{\vec{x}_k}+1)}{2}}\corr{A_{x_1}^{(1)}\ldots A_{x_{k-1}}^{(k-1)}M^{(\mathcal{D})}_{x_{k}x_{k+1}\ldots x_{N}}}\\
      =&\sum_{\substack{\vec{x}_{N-k}\in\DE{0,1}^{\times N-k}\\\omega_{\vec{x}_{N-k}} \text{ odd}}}(-1)^{\frac{\omega_{\vec{x}_{N-k}}\cdot(\omega_{\vec{x}_{N-k}}+1)}{2}} S_{k,\vec{x}_{N-k}}^- +\sum_{\substack{\vec{x}_{N-k}\in\DE{0,1}^{\times N-k}\\\omega_{\vec{x}_{N-k}} \text{ even}}}(-1)^{\frac{\omega_{\vec{x}_{N-k}}\cdot(\omega_{\vec{x}_{N-k}}+1)}{2}} S_{k,\vec{x}_{N-k}}^+
\end{split}
\end{align}
where $S_{k,\vec{x}_{N-k}}^{\pm}$ is the k-partite Svetlichny inequality with the set of dishonest parties  using the pair of observables $\DE{M^{(\mathcal{D})}_{0,\vec{x}_{N-k}},M^{(\mathcal{D})}_{1,\vec{x}_{N-k}}}$ labelled by $\vec{x}_{N-k}$. For the third equality we use the fact that
\begin{align}
\omega_{\vec{x}_{N-k}}\text{ even: }   (-1)^{{\frac{\omega_{\vec{x}_k}\cdot(\omega_{\vec{x}_k}+1)}{2}}+ \omega_{\vec{x}_k}\omega_{\vec{x}_{N-k}}}= (-1)^{{\frac{\omega_{\vec{x}_k}\cdot(\omega_{\vec{x}_k}+1)}{2}}},
\end{align}
\begin{align}
 \omega_{\vec{x}_{N-k}}\text{ odd: }   (-1)^{{\frac{\omega_{\vec{x}_k}\cdot(\omega_{\vec{x}_k}+1)}{2}}+ \omega_{\vec{x}_k}\omega_{\vec{x}_{N-k}}}= (-1)^{{\frac{\omega_{\vec{x}_k}\cdot(\omega_{\vec{x}_k}-1)}{2}}+\omega_{\vec{x}_k}(\omega_{\vec{x}_{N-k}}+1)}=(-1)^{{\frac{\omega_{\vec{x}_k}\cdot(\omega_{\vec{x}_k}-1)}{2}}}.
\end{align}

Now let $s_N=|S_N^+|$ and $s_k=\max\DE{|S^-_{k,\vec{x}_{N-k}}|,|S^+_{k,\vec{x}_{N-k}}|}$. Then we have 
\begin{align}
\begin{split}
  s_N  =&\left|\sum_{\substack{\vec{x}_{N-k}\in\DE{0,1}^{\times N-k}\\\omega_{\vec{x}_{N-k}} \text{ odd}}}(-1)^{\frac{\omega_{\vec{x}_{N-k}}\cdot(\omega_{\vec{x}_{N-k}}+1)}{2}} S_{k,\vec{x}_{N-k}}^- +\sum_{\substack{\vec{x}_{N-k}\in\DE{0,1}^{\times N-k}\\\omega_{\vec{x}_{N-k}} \text{ even}}}(-1)^{\frac{\omega_{\vec{x}_{N-k}}\cdot(\omega_{\vec{x}_{N-k}}+1)}{2}} S_{k,\vec{x}_{N-k}}^+\right|\\
 \leq &\sum_{\substack{\vec{x}_{N-k}\in\DE{0,1}^{\times N-k}\\\omega_{\vec{x}_{N-k}} \text{ odd}}}\left|S_{k,\vec{x}_{N-k}}^-\right|+\sum_{\substack{\vec{x}_{N-k}\in\DE{0,1}^{\times N-k}\\\omega_{\vec{x}_{N-k}} \text{ even}}}\left| S_{k,\vec{x}_{N-k}}^+\right|\\
 \leq&\, 2^{N-k}\max\DE{|S^-_{k,\vec{x}_{N-k}}|,|S^+_{k,\vec{x}_{N-k}}|}\\
 =&\,2^{N-k}\,s_k
\end{split}
\end{align}

This implies that there is a $k$-partite Svetlichny inequality, among the  $\DE{S^-_{k,\vec{x}_{N-k}},S^+_{k,\vec{x}_{N-k}}}$, that achieves a value $s_k\geq \frac{s_N}{2^{N-k}}$. 

A similar result can be straightforwardly derived for $S_N^-$. 
Moreover, permutational invariance of the Svetlichny inequality implies the result for an arbitrary set $\mathcal{D}$.
Therefore, when $N-k+1$ parties collaborate in order to achieve a violation $s_N$ for the $N$-partite Svetlichny inequality, the same strategy violates a $k$-partite Svetlichny inequality by at least $    s_k\geq \frac{s_N}{2^{N-k}}$.
\end{proof}

\section{Self-testing with dishonest parties}\label{app:net-selftest}

In order to prove self-testing for the Svetlichny inequality with dishonest parties, we introduce a sum of squares decomposition for the shifted Bell operator $2^{N-1} \sqrt{2} \mathds{1} - \hat{S}_{N}^+$ assuming the set of dishonest parties employ the joint observables $M^{(\mathcal{D})}_{\vec{x}_{\mathcal{D}}}$, where  $\hat{S}_N^+$ is the operator associated to the Bell inequality $S_N^+$. To do so, we use the symmetries of the Svetlichny inequality, explored in the proof of Proposition~\ref{prop:sk_violation}, in order to break it into sums of CHSH-like inequalities~\cite{CHSH}, where the inputs of $N-2$ parties is fixed. Rewriting \eqref{eq:SN+_dec} for the case {of a decomposition into $S_2^{\pm}$ terms}, one gets

\begin{align}\label{app:SNdec2}
S_N^+ =&\sum_{\substack{\vec{x}_{N-2}\in\DE{0,1}^{\times N-2}\\\omega_{\vec{x}_{N-2}} \text{ odd}}}(-1)^{\frac{\omega_{\vec{x}_{N-2}}\cdot(\omega_{\vec{x}_{N-2}}+1)}{2}} S_{2,\vec{x}_{N-2}}^- +\sum_{\substack{\vec{x}_{N-2}\in\DE{0,1}^{\times N-2}\\\omega_{\vec{x}_{N-2}} \text{ even}}}(-1)^{\frac{\omega_{\vec{x}_{N-2}}\cdot(\omega_{\vec{x}_{N-2}}+1)}{2}} S_{2,\vec{x}_{N-2}}^+,
\end{align}
where we can now write explicitly
\begin{equation}
S_{2,\vec{x}_{N-2}}^+ = \langle A_{0}^{(1)} A_{0}^{(2)}  \ldots A_{x_{k-1}}^{(k-1)}M_{\vec{x}_{\mathcal{D}}}^{(\mathcal{D})} \rangle 
- \langle A_{0}^{(1)} A_{1}^{(2)}\ldots A_{x_{k-1}}^{(k-1)}M_{\vec{x}_{\mathcal{D}}}^{(\mathcal{D})} \rangle    
- \langle A_{1}^{(1)} A_{0}^{(2)} \ldots A_{x_{k-1}}^{(k-1)}M_{\vec{x}_{\mathcal{D}}}^{(\mathcal{D})}\rangle    
-\langle A_{1}^{(1)} A_{1}^{(2)} \ldots A_{x_{k-1}}^{(k-1)}M_{\vec{x}_{\mathcal{D}}}^{(\mathcal{D})} \rangle    ,
\end{equation}
and
\begin{equation}
S_{2,\vec{x}_{N-2}}^- = \langle A_{0}^{(1)} A_{0}^{(2)} \ldots A_{x_{k-1}}^{(k-1)}M_{\vec{x}_{\mathcal{D}}}^{(\mathcal{D})} \rangle 
+ \langle A_{0}^{(1)} A_{1}^{(2)} \ldots A_{x_{k-1}}^{(k-1)}M_{\vec{x}_{\mathcal{D}}}^{(\mathcal{D})} \rangle    
+ \langle A_{1}^{(1)} A_{0}^{(2)} \ldots A_{x_{k-1}}^{(k-1)}M_{\vec{x}_{\mathcal{D}}}^{(\mathcal{D})} \rangle    
- \langle A_{1}^{(1)} A_{1}^{(2)} \ldots A_{x_{k-1}}^{(k-1)}M_{\vec{x}_{\mathcal{D}}}^{(\mathcal{D})} \rangle   .
\end{equation}
Note that, as opposed to the previous Appendix, here we are taking a decomposition that involves both trusted and untrusted parties. That is, while we keep having $N -k + 1$ untrusted parties, with $k$ arbitrary, the decomposition 
\eqref{app:SNdec2} is taken to be according to the $2|N-2$ partition independently from the value of $k$. One can also see that the inequalities $S_{2,\vec{x}_{N-2}}^\pm$ are simply two variations of CHSH, played between party $1$ and $2$, while the settings of the other parties are fixed. We can then write a SOS for both of these inequalities 

\begin{align}
2 \sqrt{2} \mathds{1} - \hat{S}_{2,\vec{x}_{N-2}}^+ & =  \frac{1}{\sqrt{2}} \left[ \left( \mathds{1} - \frac{\hat{A}_0^{(1)} - \hat{A}_1^{(1)}}{\sqrt{2}} \hat{A}_0^{(2)} \ldots A_{x_{k-1}}^{(k-1)}M_{\vec{x}_{\mathcal{D}}}^{(\mathcal{D})} \right)^2 
+ \left( \mathds{1} + \frac{\hat{A}_0^{(1)} + \hat{A}_1^{(1)}}{\sqrt{2}} \hat{A}_1^{(2)} \ldots A_{x_{k-1}}^{(k-1)}M_{\vec{x}_{\mathcal{D}}}^{(\mathcal{D})} \right)^2  \right] \, \, , \\
2 \sqrt{2} \mathds{1} - \hat{S}_{2,\vec{x}_{N-2}}^- & =  \frac{1}{\sqrt{2}} \left[ \left( \mathds{1} - \frac{\hat{A}_0^{(1)} + \hat{A}_1^{(1)}}{\sqrt{2}} \hat{A}_0^{(2)} \ldots A_{x_{k-1}}^{(k-1)}M_{\vec{x}_{\mathcal{D}}}^{(\mathcal{D})} \right)^2 
+ \left( \mathds{1} - \frac{\hat{A}_0^{(1)} - \hat{A}_1^{(1)}}{\sqrt{2}} \hat{A}_1^{(2)} \ldots A_{x_{k-1}}^{(k-1)}M_{\vec{x}_{\mathcal{D}}}^{(\mathcal{D})} \right)^2  \right] \, \, ,
\end{align}
and combine them to get a SOS for the whole inequality
\begin{align}\label{eq:SOS_SNdishonest}
\begin{split}
2^{N-1} \sqrt{2} \mathds{1} - \hat{S}_{N}^+  =&
\frac{1}{\sqrt{2}} \sum_{\substack{\vec{x}_{N-2}\in\DE{0,1}^{\times N-2}\\\omega_{\vec{x}_{N-2}} \text{ odd}}}  \left[ \left( \mathds{1}  - (-1)^{\frac{\omega_{\vec{x}_{N-2}}\cdot(\omega_{\vec{x}_{N-2}}+1)}{2}} \frac{{\hat{A}}_0^{(1)} + {\hat{A}}_1^{(1)}}{\sqrt{2}} {\hat{A}}_0^{(2)} \ldots A_{x_{k-1}}^{(k-1)}M_{\vec{x}_{\mathcal{D}}}^{(\mathcal{D})} \right)^2 \right.  \\ 
 & + \left. \left( \mathds{1} - (-1)^{\frac{\omega_{\vec{x}_{N-2}}\cdot(\omega_{\vec{x}_{N-2}}+1)}{2}} \frac{{A}_0^{(1)} - {{A}}_1^{(1)}}{\sqrt{2}} {\hat{A}}_1^{(2)} \ldots A_{x_{k-1}}^{(k-1)}M_{\vec{x}_{\mathcal{D}}}^{(\mathcal{D})} \right)^2  \right]  \\
 &+   \frac{1}{\sqrt{2}} \sum_{\substack{\vec{x}_{N-2}\in\DE{0,1}^{\times N-2}\\\omega_{\vec{x}_{N-2}} \text{ even}}} \left[ \left( \mathds{1} -(-1)^{\frac{\omega_{\vec{x}_{N-2}}\cdot(\omega_{\vec{x}_{N-2}}+1)}{2}}\frac{{A}_0^{(1)} - {{A}}_1^{(1)}}{\sqrt{2}} {{A}}_0^{(2)} \ldots A_{x_{k-1}}^{(k-1)}M_{\vec{x}_{\mathcal{D}}}^{(\mathcal{D})} \right)^2 \right. \\
 &  + \left. \left( \mathds{1} + (-1)^{\frac{\omega_{\vec{x}_{N-2}}\cdot(\omega_{\vec{x}_{N-2}}+1)}{2}}\frac{{\hat{A}}_0^{(1)} + {{A}}_1^{(1)}}{\sqrt{2}} {{A}}_1^{(2)} \ldots A_{x_{k-1}}^{(k-1)}M_{\vec{x}_{\mathcal{D}}}^{(\mathcal{D})} \right)^2  \right] \, .
 \end{split}
\end{align}

This indeed proves that the maximal violation of  ${S}_{N}^+$ achievable by quantum correlations considering the global action of dishonest parties  is also $2^{N-1} \sqrt{2}$. Moreover, the state achieving it satisfies to following conditions
\begin{eqnarray}
(-1)^{\frac{\omega_{\vec{x}_{N-2}}\cdot(\omega_{\vec{x}_{N-2}}+1)}{2}} \frac{{A}_0^{(1)} + {A}_1^{(1)}}{\sqrt{2}} {A}_0^{(2)} \ldots A_{x_{k-1}}^{(k-1)}M_{\vec{x}_{\mathcal{D}}}^{(\mathcal{D})} \ket{\psi} & = & \ket{\psi} \quad \omega_{\vec{x}_{N-2}} \text{ odd}, \label{app:SOSodd1}\\
(-1)^{\frac{\omega_{\vec{x}_{N-2}}\cdot(\omega_{\vec{x}_{N-2}}+1)}{2}} \frac{{A}_0^{(1)} - {A}_1^{(1)}}{\sqrt{2}} {A}_1^{(2)} \ldots A_{x_{k-1}}^{(k-1)}M_{\vec{x}_{\mathcal{D}}}^{(\mathcal{D})}\ket{\psi} & = & \ket{\psi} \quad \omega_{\vec{x}_{N-2}} \text{ odd}, \label{app:SOSodd2} \\ 
(-1)^{\frac{\omega_{\vec{x}_{N-2}}\cdot(\omega_{\vec{x}_{N-2}}+1)}{2}}
\frac{{A}_0^{(1)} - {A}_1^{(1)}}{\sqrt{2}} {A}_0^{(2)} \ldots A_{x_{k-1}}^{(k-1)}M_{\vec{x}_{\mathcal{D}}}^{(\mathcal{D})} \ket{\psi} & = & \ket{\psi} \quad \omega_{\vec{x}_{N-2}} \text{ even}, \label{app:SOSeven1} \\
(-1)^{\frac{\omega_{\vec{x}_{N-2}}\cdot(\omega_{\vec{x}_{N-2}}+1)}{2}}
\frac{{A}_0^{(1)} + {A}_1^{(1)}}{\sqrt{2}} {A}_1^{(2)} \ldots A_{x_{k-1}}^{(k-1)}M_{\vec{x}_{\mathcal{D}}}^{(\mathcal{D})} \ket{\psi} & = & - \ket{\psi} \quad \omega_{\vec{x}_{N-2}} \text{ even}. \label{app:SOSeven2}
\end{eqnarray}

We will now show that by introducing the substitutions
\begin{eqnarray}\label{app:SNsubs}
\frac{{A}_0^{(1)} + {A}_1^{(1)}}{\sqrt{2}}  &= &   - \hat{Z}_1 \quad , \quad  \frac{{A}_0^{(1)} - {A}_1^{(1)}}{\sqrt{2}}  =  \hat{X}_1 \nonumber \\
{A}_0^{(i)} & = &  \hat{Z}_i \quad , \quad  {A}_1^{(i)} =   \hat{X}_i  \, \, , \quad  i = 2, \ldots, {k-1}  \, \, , \nonumber \\
{M}^{(\mathcal{D})}_{0, \vec{x}_{N-k}} & = &  (-1)^{\frac{\omega_{\vec{x}_{N-k}} (\omega_{\vec{x}_{N-k}} + 1)}{2}}\hat{Z}_{k} \quad , \quad {M}^{(\mathcal{D})}_{1, \vec{x}_{N-k}} =   -(-1)^{\frac{(\omega_{\vec{x}_{N-k}} +1) (\omega_{\vec{x}_{N-k}} + 2)}{2}} \hat{X}_{k}  \, \, , \quad \omega_{\vec{x}_{N-k}} \quad \text{even}  \\
{M}^{(\mathcal{D})}_{0, \vec{x}_{N-k}} & = &  -(-1)^{\frac{\omega_{\vec{x}_{N-k}} (\omega_{\vec{x}_{N-k}} + 1)}{2}} \hat{X}_{k} \quad , \quad {M}^{(\mathcal{D})}_{1, \vec{x}_{N-k}} =   (-1)^{\frac{(\omega_{\vec{x}_{N-k}} +1) (\omega_{\vec{x}_{N-k}} + 2)}{2}} \hat{Z}_{k}  \, \, , \quad \omega_{\vec{x}_{N-k}} \quad \text{odd} \nonumber
\end{eqnarray}
some of the conditions arising from the SOS correspond to the stabilizing conditions of the generators of a fully-connected graph states of $k$ {parties} \cite{hein2006entanglement}. That is, equations of the form 
$\hat{S}_i \ket{\psi} = \ket{\psi}$ for all $i = 1,\ldots,k$, where 
\begin{equation}
 \hat{S}_i = \hat{X}_i \bigotimes_{j \in [k]/i} \hat{Z}_j
\end{equation}
Let us see that explicitly: take an arbitrary input string $\vec{x}_{N-k}$ for the untrusted parties $k+1,\ldots,N$. First, consider conditions \eqref{app:SOSeven1}, \eqref{app:SOSeven2} for the string choice $\vec{x}_{N-2} = 0\ldots00\vec{x}_{N-k}$ if $\omega_{\vec{x}_{N-k}}$ is even and string choice  $\vec{x}_{N-2} = 0\ldots01\vec{x}_{N-k}$ if odd. The corresponding Hamming weight reads $\omega_{\vec{x}_{N-2}} = \omega_{\vec{x}_{N-k}}$ ($\omega_{\vec{x}_{N-k}} +1$) when $\omega_{\vec{x}_{N-k}}$ is even (odd). In both cases, $\omega_{\vec{x}_{N-2}}$ is indeed an even number. Written with the substitutions \eqref{app:SNsubs}, the two conditions become
\begin{eqnarray}
\hat{X}_1 \hat{Z}_2 \hat{Z}_3 \ldots \hat{Z}_k \ket{\psi}& = \ket{\psi} \label{app:S1} \\
- \hat{Z}_1 \hat{X}_2 \hat{Z}_3 \ldots \hat{Z}_k\ket{\psi}& = - \ket{\psi} \label{app:S2}
\end{eqnarray}
which correspond to the stabilising conditions for $\hat{S}_1$ and $\hat{S}_2$. Then, consider all the strings $\vec{y}_{k-3}$ of length $k-3$ with Hamming weight $1$ (that is, where all except one entry are $0$). We take the condition \eqref{app:SOSodd1} for the strings $\vec{x}_{N-2} = 00\vec{y}_{k-3}0\vec{x}_{N-k}$ if $\omega_{\vec{x}_{N-k}}$ is even and string choice  $\vec{x}_{N-2} =  00\vec{y}_{k-3}1\vec{x}_{N-k}$ if odd. The corresponding Hamming weight reads $\omega_{\vec{x}_{N-2}} = \omega_{\vec{x}_{N-k}}+1$ ($\omega_{\vec{x}_{N-k}} +2$) when $\omega_{\vec{x}_{N-k}}$ is even (odd).
Hence, we fulfill the required condition of it being an odd number. For each string $\vec{y}_{k-3}$ and after applying the substitutions,  \eqref{app:SOSodd1} reads
\begin{eqnarray}\label{app:S3}
     (-1)^{(\omega_{\vec{x}_{N-k}}+1)} (-\hat{Z}_1) \hat{Z}_2 \ldots \hat{Z}_{i-1} \hat{X}_i  \hat{Z}_{i+1} \ldots \hat{Z}_k \ket{\psi} &=& \ket{\psi} \quad  \omega_{\vec{x}_{N-k}} \quad \text{even} \nonumber \\
    (-1)^{(\omega_{\vec{x}_{N-k}}+2)} (-\hat{Z}_1) \hat{Z}_2 \ldots \hat{Z}_{i-1} \hat{X}_i  \hat{Z}_{i+1} \ldots \hat{Z}_k \ket{\psi} &=& \ket{\psi} \quad  \omega_{\vec{x}_{N-k}} \quad \text{odd} \, ,
\end{eqnarray}
where $i = 3,\ldots,k-1$ corresponds to the position where a $1$ appears in the string $00\vec{y}_{k-3}$. In both cases the minus signs compensate each other and the equations \eqref{app:S3} 
recover the stabilising conditions $\hat{S}_i$ for $i = 3,\ldots,k-1$. 
We are now left with the last stabilising condition to be recovered. For that, we consider again equation \eqref{app:SOSodd1}, with string choice $\vec{x}_{N-2} = 0\ldots01\vec{x}_{N-k}$ if $\omega_{\vec{x}_{N-k}}$ is even and string choice  $\vec{x}_{N-2} = 0\ldots00\vec{x}_{N-k}$ if odd. The corresponding Hamming weight reads $\omega_{\vec{x}_{N-2}} = \omega_{\vec{x}_{N-k}} +1$ ($\omega_{\vec{x}_{N-k}}$) when $\omega_{\vec{x}_{N-k}}$ is even (odd). In both cases, $\omega_{\vec{x}_{N-2}}$ is an odd number as required. After applying once again the substitutions \eqref{app:SNsubs}, the condition becomes
\begin{equation}
(-\hat{Z_1}) \hat{Z_2} \ldots \hat{Z}_{k-1} (- \hat{X}_k) \ket{\psi} = \ket{\psi}  \, ,
\end{equation}
which is exactly the stabilising condition for $\hat{S}_k$.
Hence, we have recovered all the stabilizing conditions for the $N$ generators $\hat{S}_i = \hat{X}_i \bigotimes_{j \in [k]/i} \hat{Z}_j$ for $i = 1,\ldots,k$. 

Recall now that a fully-connected graph state is local-unitary equivalent to a GHZ state \cite{hein2006entanglement}. From here on, we can proceed as done in \cite{baccari2020scalable} and prove both measurement and state self-testing from the maximal violation of the Svetlichny inequality. For the sake of clarity, we will revise here the main steps of the proof, stating them in the framework of a dishonest-parties scenario.

\textit{Step 1: Operator properties.} Notice that the operators $\hat{X}_1, \hat{Z}_1$ anticommute by construction, while $\hat{X}_i^2 = \hat{Z}_i^2 = \mathds{1}$ for all $i \neq 1$. By following the same steps as in \cite{baccari2020scalable}, one can exploit the stabilising conditions \eqref{app:S1},\eqref{app:S2},\eqref{app:S3} to show that
\begin{eqnarray}\label{app:qubitprop}
\lbrace \hat{X}_i ,  \hat{Z}_i \rbrace \ket{\psi}& = \ket{\psi} \\
\hat{Z}_i^2 \ket{\psi} = \hat{X}_i^2 \ket{\psi} & = \ket{\psi}
\end{eqnarray}
for all $i = 1,\ldots, k$. Notice that from here on we are defining the operators $\hat{X}_k$, $\hat{Z}_k$ by fixing a specific choice of input string $\vec{x}_{N-k}$.

\textit{Step 2: SWAP isometry with dishonest parties.} Let us now define the isometries mapping the measured state to the target GHZ state. For that, we introduce a modification of the so-called SWAP isometry that is apt to a scenario with dishonest parties. The output of the SWAP isometry reads as follows
\begin{equation}\label{app:iso}
\Lambda \left(\ket{+}^{\otimes k}\otimes \ket{\psi}\right)     
=\sum_{\tau \in \{0,1\}^k}\ket{\tau}\otimes \left( \bigotimes_{j=1}^k\hat{X}_j^{\tau_j}\hat{Z}_j^{(\tau_j)} \right) \ket{\psi},
\end{equation}
where $\hat{X}_i$ and $\hat{Z}_i$ are those defined in \eqref{app:SNsubs} and we have also defined $\hat{Z}_i^{(\tau_j)} = [\mathbbm{1}+(-1)^{\tau_j}\hat{Z}_j]/2$, while the summation is over all $k$-element sequences $(\tau_1,\ldots,\tau_k)$ with each $\tau_i\in\{0,1\}$.
The action of this isometry is to perform a unitary operation $\Lambda =\Lambda_1\otimes\ldots \Lambda_{k-1} \otimes\Lambda_{\mathcal{D}}$
on the state $\ket{+}^{\otimes k}\otimes \ket{\psi}$. Each unitary $\Lambda_i$ for $i \leq k-1$ acts on the $i$-th particle of $\ket{\psi}$ and one of the qubits in the state $\ket{+}$, while the last unitary $\Lambda_{\mathcal{D}}$ acts globally on the $N-k-1$ particles of the dishonest parties and a single qubit state $\ket{+}$. This follows from the fact that the operators $\hat{X_k}$, $\hat{Z_k}$ are defined by involving measurement operators of all dishonest parties. A visual representation of a local branch and the global branch of the isometry $\Lambda$ is shown in Fig. \ref{fig:swapp}.
\begin{figure}[h!]
\centering
\includegraphics[width=0.8\textwidth]{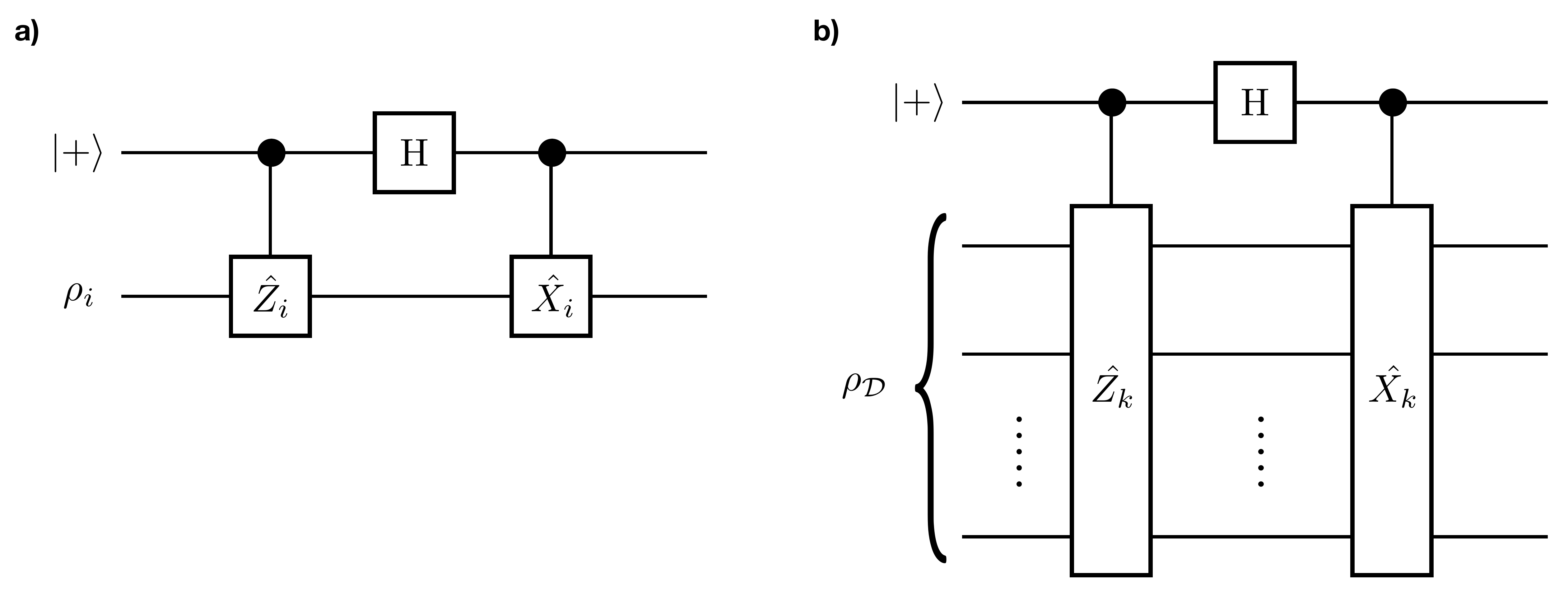}
\caption{a) A local branch $\Lambda_i$ of the SWAP gate. Each local branch acts on the subsystem of $\ket{\psi}$ in possession of the $i$-th honest party and one ancillary qubit in the state $\ket{+}$. b) The global branch $\Lambda_{\mathcal{D}}$ of the SWAP gate, acting jointly on the subsystems of $\ket{\psi}$ shared by set of $N-k-1$ dishonest parties and a single ancillary qubit in the state $\ket{+}$.} \label{fig:swapp}
\end{figure}

\textit{Step 3: Self-testing statements.} With the SWAP isometry defined in \eqref{app:iso}, we can now proceed in exactly the same way as in  \cite{baccari2020scalable}. By applying a properly chosen sequence of substitutions coming from the properties \eqref{app:qubitprop} and the stabilising conditions \eqref{app:S1},\eqref{app:S2},\eqref{app:S3}, one can show 
\begin{eqnarray}\label{iso1}
 \Lambda\left(\ket{+}^{\otimes k}\otimes \ket{\psi}\right) &=&\sum_{\tau \in \{0,1\}^k}(-1)^{n(\tau)}\ket{\tau}\otimes\left( \bigotimes_{j=1}^k Z_j^{(0)} \right) \ket{\psi} \nonumber\\ 
&&= \ket{\psi_G}\otimes\ket{\zeta},
\end{eqnarray}
where $\ket{\zeta} = \left( \bigotimes_{j=1}^k Z_j^{(0)} \right) \ket{\psi}$ represent some uncorrelated degrees of freedom and
\begin{equation}
\ket{\psi_G} = \sum_{\tau \in \{0,1\}^k}(-1)^{n(\tau)}\ket{\tau}    
\end{equation}
is exactly the expression of the $k$-partite fully connected graph state in the computational basis \cite{hein2006entanglement}, $n(\tau)$ counting the number of edges connecting qubits which are in state $\ket{1}$ in ket $\ket{\tau}$. Since the GHZ state is locally-unitary equivalent to $\ket{\psi_G}$, we have just shown the existence of an isometry of the form $\Lambda_1\otimes\ldots \Lambda_{k-1} \otimes\Lambda_{\mathcal{D}}$ mapping the state maximally violating the Svetlichny inequality {with dishonest parties} to a $k$-partite GHZ state.
With an analogous procedure, one can show measurement self-testing, namely 
\begin{eqnarray}\label{isomeas}
 \Lambda\left(\ket{+}^{\otimes k}\otimes A_{x_1}^{(1)}\otimes \ldots \otimes A_{x_{k-1}}^{(k-1)} \otimes M_{\vec{x}_\mathcal{D}}^{(\mathcal{D})}  \ket{\psi}\right) = \bar{A}_{x_1}^{(1)}\otimes \ldots \otimes \bar{A}_{x_{k-1}}^{(k-1)} \otimes \bar{A}_{\vec{x}_\mathcal{D}}^{(k)} \ket{\psi_G}\otimes\ket{\zeta},
\end{eqnarray}
where the states $\ket{\psi_G}$ and $\ket{\zeta}$ are defined as before and the target measurements are
\begin{eqnarray}\label{app:SNmeas}
 \bar{A}_0^{(1)}&= &    \frac{\sigma_x - \sigma_z}{\sqrt{2}}  \quad , \quad   \bar{A}_1^{(1)} = -   \frac{\sigma_x + \sigma_z}{\sqrt{2}}  \nonumber \\
\bar{A}_0^{(i)} & = &  \sigma_z \quad , \quad  \bar{A}_1^{(i)} =   \sigma_x \, \, , \quad  i = 2, \ldots, {k-1}  \, \, , \nonumber \\
\bar{A}^{(k)}_{0, \vec{x}_{N-k}} & = &  (-1)^{\frac{\omega_{\vec{x}_{N-k}} (\omega_{\vec{x}_{N-k}} + 1)}{2}}\sigma_z \quad , \quad \bar{A}^{(k)}_{1, \vec{x}_{N-k}} =   -(-1)^{\frac{(\omega_{\vec{x}_{N-k}} +1) (\omega_{\vec{x}_{N-k}} + 2)}{2}} \sigma_x  \, \, , \quad \omega_{\vec{x}_{N-k}} \quad \text{even}  \\
\bar{A}^{(k)}_{0, \vec{x}_{N-k}} & = &  -(-1)^{\frac{\omega_{\vec{x}_{N-k}} (\omega_{\vec{x}_{N-k}} + 1)}{2}} \sigma_x \quad , \quad \bar{A}^{(k)}_{1, \vec{x}_{N-k}} =   (-1)^{\frac{(\omega_{\vec{x}_{N-k}} +1) (\omega_{\vec{x}_{N-k}} + 2)}{2}} \sigma_z  \, \, , \quad \omega_{\vec{x}_{N-k}} \quad \text{odd} \nonumber
\end{eqnarray}
with $\sigma_x,\sigma_z$ being the qubit Pauli operators. Notice that the correlations obtained by performing the above choice of measurements on the state $\ket{\psi_G}$ maximally violates a $k$-partite Svetlichny inequality for any pair of operators for the grouped dishonest parties associated to a specific input string $\vec{x}_{N-k}$.
Notice that, since the last $N-k+1$ parties are untrusted, the only measurement self-testing statement we can make about them is that the operators $\hat{A}^{(k)}_{0 \vec{x}_{N-k}}$, $\hat{A}^{(k)}_{1 \vec{x}_{N-k}}$  behave as $\pm \sigma_z, \pm \sigma_x$ measurements on the $k$-th party. Hence we can not deduce whether the $N-k+1$ parties are actually gathered together, simply measuring the two Pauli operators on a single qubit, or they are actually different space-like-separated parties performing local Pauli measurements on a $N$-partite GHZ state. 

Lastly, if we take the special case of {$k = N$, namely at most one dishonest party}, the above self-testing result implies that the maximal violation of the Svetlichy inequality self-tests a $N$-partite GHZ state
and the corresponding measurements, in the standard definition of self-testing. 

{The above proves Theorem~\ref{thm:self-test} for $S_N^+$. The result for $S_N^-$ follows straightforwardly if we note that $S_N^-$ can be obtained from $S_N^+$ by taking $A_1^{(i)}\rightarrow -A_1^{(i)}\, \forall i$, which corresponds to rellabling the outcomes of $A_1^{(i)}$:
\begin{align}
\begin{split}   S_N^{-}\rightarrow&\sum_{\vec{x}}(-1)^{\frac{w_{\vec{x}}(w_{\vec{x}}- 1)}{2}}(-1)^{w_{\vec{x}}}\corr{A_{x_1}^{(1)}A_{x_2}^{(2)}\ldots A_{x_N}^{(N)}}\\
    &= \sum_{\vec{x}}(-1)^{\frac{w_{\vec{x}}(w_{\vec{x}}+ 1)}{2}}\corr{A_{x_1}^{(1)}A_{x_2}^{(2)}\ldots A_{x_N}^{(N)}}\\
    &=S_N^+.
    \end{split}
\end{align}
}

\section{Numerical evidence for robust self-testing with dishonest parties}\label{app:robust}

Theorem~\ref{thm:F_D_k} shows that  bounds on the device-independent fidelity with dishonest parties can be obtained from bounds for the device-independent fidelity in the standard Bell scenario. Therefore, our task reduces to finding coefficients $f_k$ and $\mu_k$ that bound the device-independent fidelity as a function of the violation of the $k$-partite Svetlichny inequality in the standard scenario, i.e. we want to bound
\begin{align}\label{eq:Fstandard}
F_{\rm DI}(s_k)\geq f_ks_k-\mu_k,
\end{align}
where
\begin{align}\label{eq:defF}
    F_{\rm DI}(s_k)=\inf_{{\rho}\in \mathcal{S}({s}_k)} \max_{\DE{\Lambda_i}_{i=1}^k}F(\Lambda_{1}\otimes\ldots \Lambda_{k}({\rho}),\Phi^N),
\end{align}
with $\Lambda_i$, $i\in [k]$, being a quantum channel, and the infimum taken over all states that achieve violation $s_k$ for the $k$-partite Svetlichny inequality.

In order to find such bounds, we employ the self-testing from operator inequality (STOPI) method introduced in~\cite{JedSelftest} and further developed in~\cite{Tim-st}. 
The method consists in choosing a particular extractable channel, i.e. fixing the maps $\DE{\Lambda_i}_{i=1}^k$ in \eqref{eq:defF}, and turning \eqref{eq:Fstandard} into an operator inequality that is valid for any quantum state. This is done by noting that for fixed quantum state and channels, the r.h.s. of \eqref{eq:defF} can be turned into
\begin{align}
\begin{split}
    F(\otimes_{i=1}^k \Lambda_i(\rho),\Phi_k)&= \Tr\de{\otimes_{i=1}^k \Lambda_i(\rho)\Phi_k}\\
    &=\Tr\de{\otimes_{i=1}^k \Lambda^\dag_i(\Phi_k)\rho}\\
    &=\Tr\de{K\rho}
    \end{split}
\end{align}
where $K:=\otimes_{i=1}^k \Lambda^\dag_i(\Phi_k)$.
Moreover $s_k=\Tr\de{W\rho}$, where $W$ is the Bell operator that represents the Svetlichny inequality. Therefore a valid bound of the form \eqref{eq:Fstandard} is obtained for every $f_k$
, $\mu_k$ such that following operator
\begin{align}\label{eq:operator}
K-f_kW+\mu_k\id
\end{align}
is positive.
For more details we refer the reader to \cite{Tim-st}.

For $k=2$ the Svetlichny inequality reduces to the CHSH inequality. An analytical lower bound of the form \eqref{eq:Fstandard} has already been proved in~\cite{JedSelftest}, namely
\begin{align}
   F_{\rm DI}(s_2)\geq \frac{4+5\sqrt{2}}{16}s_2-\frac{1+2\sqrt{2}}{4}.
\end{align}

Using the STOPI method of~\cite{JedSelftest,Tim-st} we obtained numerical evidence to conjecture the following analytical results for $k=3,4$:
\begin{align}
    F_{\rm DI}(s_3)&\geq \frac{3(1+\sqrt{2})}{16}s_3-\frac{2+3\sqrt{2}}{4},
\\
    F_{\rm DI}(s_4)&\geq \frac{(1+\sqrt{2})}{16}s_4-\frac{\sqrt{2}}{2}.
\end{align}

In order to derive these bounds, we first note that the Svetlichny inequality consists of binary inputs and outputs, and therefore the analysis can be restricted to qubits and projective measurements~ (see \cite{JedSelftest}). Next, we fix the same extractable channels of refs.~\cite{JedSelftest}, i.e., for each party $i$, $\Lambda_i^\dag$ is a dephasing channel determined by the angle $a_i$ between the observables used for the Svetlichny violation (the observables that define the operator $W$, see eqs. (8) and (9) in ref. \cite{JedSelftest}). A lower bound on $F_{\rm DI}$ can then be obtained by (i) fixing $\mu_s$ as a function of $f_s$ in order to obtain fidelity one for maximal violation, then (ii) minimize \eqref{eq:operator} over the angles between the observables used by the parties for each value of $f_s$, and finally (iii) taking $f_s$ to be the threshold value for which the minimum eigenvalue of \eqref{eq:operator} becomes positive. 

In order to obtain the analytical conjectures, for $k=3$, we numerically computed the minimum eigenvalue of \eqref{eq:operator} with different fixed values of $f_3$. We noted that, for $f_3<f_3^{opt}$, the minimum  is negative and achieved for the angles $\DE{(\pi/2,0,0),(0,0,0)}$ (and the respective permutations). And for $f_3>f_3^{opt}$, the minimum is achieved at $(\pi/4,\pi/4,\pi/4)$ and equals zero. So the analytical value $f_3^{opt}=\frac{3(1+\sqrt{2})}{16}$ is obtained by $f_3$ such that the minimum eigenvalue of \eqref{eq:operator} for the angles  $\DE{(\pi/2,0,0),(0,0,0)}$ becomes zero. The derived analytical bound was subsequently checked using a grid, as previously done in~\cite{Tim-st} for another Bell inequality.

For the case $k=4$, numerical evidence shows that minimum eigenvalue of \eqref{eq:operator} is achieved at $(\pi/4,\pi/4,\pi/4, d)$ (and permutations), with $d$ depending on $f_4$, for $f_4<f_4^{opt}$. And for $f_4>f_4^{opt}$, the minimum is achieved at $(\pi/4,\pi/4,\pi/4,\pi/4)$ and equals zero. Similarly, the analytical value $f_4^{opt}=\frac{1+\sqrt{2}}{16}$ is determined by imposing that the minimum eigenvalue of \eqref{eq:operator} for the angles $(\pi/4,\pi/4,\pi/4, d)$ becomes zero. The result was also numerically checked using a grid.

\section{Networks with clusters}\label{app:clusters}

In order to show that our results extend to the cluster scenario, we start by generalizing Proposition~\ref{prop:sk_violation} to this scenario.

\begin{propA}\label{prop:appsk_cluster}
If the $N$-parties in the network achieve a value $s_N$ for the $N$-partite Svetlichny inequality, then  the same strategy  achieves value $s_{k}$ for a $k$-partite Svetlichny inequality where  the parties are grouped in $k$ different subsets, $\mathcal{S}_1 \ldots \mathcal{S}_k$, and the  outcome of each subset $\mathcal{S}_j$ is given by $a_j'=\bigoplus_{i\in\mathcal{S}_j}a_i$, with
\begin{align}
    s_{k}\geq \frac{s_N}{2^{N-k}}.
\end{align}
\end{propA}

\begin{proof}
We now consider the scenario where different disjoint subsets of  the $N$ parties are collaborating. Specifically let's consider that the $N$ parties are divided into $k$ disjoint subsets $\mathcal{S}_1 \ldots \mathcal{S}_k$ of size $N_1, N_2, \ldots N_k$ respectively, such that
$N_1+ N_2+ \ldots+ N_k=N$. We can then associate a global observable to the joint strategy performed by subset $\mathcal{S}_i$
\begin{align}\label{eq:clusterObser}
    A_{x_{\mathcal{S}_{i,1}}}^{(\mathcal{S}_{i,1})}    \ldots    A_{x_{\mathcal{S}_{i,N_i}}}^{(\mathcal{S}_{i,N_i})}\rightarrow M^{(\mathcal{S}_i)}_{x_{\mathcal{S}_{i,1}}\ldots x_{\mathcal{S}_{i,N_i}}}
\end{align}
and let us split the string of outcome labels into the following two vectors
\begin{align}
    \vec{x}_k=(x_{\mathcal{S}_{1,1}},x_{\mathcal{S}_{2,1}},\ldots,x_{\mathcal{S}_{k,1}})
\end{align}
\begin{align}
    \vec{x}_{N-k}=(x_{\mathcal{S}_{1,2}},\ldots,x_{\mathcal{S}_{1,N_1}},x_{\mathcal{S}_{2,2}},\ldots,x_{\mathcal{S}_{2,N_2}},\ldots x_{\mathcal{S}_{k,2}},\ldots,x_{\mathcal{S}_{k,N_k}})\label{eq:xn-k} \, ,
\end{align}
{where we numbered the parties of each subset $\mathcal{S}_i$ as $\lbrace \mathcal{S}_{i,1},\ldots, \mathcal{S}_{i, N_i} \rbrace$}.

The Svetlichny inequality can then be re-written as
\begin{align} \label{eq:SN+_dec_cluster}
    S_N^+=&\sum_{\vec{x}\in\DE{0,1}^{\times N}} (-1)^{\frac{\omega_{\vec{x}}\cdot(\omega_{\vec{x}}+1)}{2}}\corr{M^{(\mathcal{S}_1)}_{x_{\mathcal{S}_{1,1}}\ldots x_{\mathcal{S}_{1,N_1}}}\ldots M^{(\mathcal{S}_k)}_{x_{{\mathcal{S}_{k,1}}}\ldots x_{\mathcal{S}_{k,N_k}}}}\\
    =&\sum_{\vec{x}_{N-k}\in\DE{0,1}^{\times N-k}}\sum_{\vec{x}_{k}\in\DE{0,1}^{\times k}}(-1)^{\frac{(\omega_{\vec{x}_k}+\omega_{\vec{x}_{N-k}})\cdot(\omega_{\vec{x}_k}+\omega_{\vec{x}_{N-k}}+1)}{2}}\corr{M^{(s_1)}_{x_{s_{1,1}}x_{s_{1,2}}\ldots x_{s_{1,N_1}}}\ldots M^{(s_k)}_{x_{s_{k,1}}x_{s_{k,2}}\ldots x_{s_{k,N_k}}}}\\
     =&\sum_{\vec{x}_{N-k}\in\DE{0,1}^{\times N-k}} (-1)^{\frac{\omega_{\vec{x}_{N-k}}\cdot(\omega_{\vec{x}_{N-k}}+1)}{2}}\sum_{\vec{x}_{k}\in\DE{0,1}^{\times k}}(-1)^{\frac{\omega_{\vec{x}_k}\cdot(\omega_{\vec{x}_k}+1)}{2}+\omega_{\vec{x}_{k}}\omega_{\vec{x}_{N-k}}}\corr{M^{(s_1)}_{x_{s_{1,1}}x_{s_{1,2}}\ldots x_{s_{1,N_1}}}\ldots M^{(s_k)}_{x_{s_{k,1}}x_{s_{k,2}}\ldots x_{s_{k,N_k}}}}\\
      =&\sum_{\substack{\vec{x}_{N-k}\in\DE{0,1}^{\times N-k}\\\omega_{\vec{x}_{N-k}} \text{ odd}}} (-1)^{\frac{\omega_{\vec{x}_{N-k}}\cdot(\omega_{\vec{x}_{N-k}}+1)}{2}}\sum_{\vec{x}_{k}\in\DE{0,1}^{\times k}}(-1)^{\frac{\omega_{\vec{x}_k}\cdot(\omega_{\vec{x}_k}-1)}{2}}\corr{M^{(s_1)}_{x_{s_{1,1}}x_{s_{1,2}}\ldots x_{s_{1,N_1}}}\ldots M^{(s_k)}_{x_{s_{k,1}}x_{s_{k,2}}\ldots x_{s_{k,N_k}}}}\\
      &+\sum_{\substack{\vec{x}_{N-k}\in\DE{0,1}^{\times N-k}\\\omega_{\vec{x}_{N-k}} \text{ even}}} (-1)^{\frac{\omega_{\vec{x}_{N-k}}\cdot(\omega_{\vec{x}_{N-k}}+1)}{2}}\sum_{\vec{x}_{k}\in\DE{0,1}^{\times k}}(-1)^{\frac{\omega_{\vec{x}_k}\cdot(\omega_{\vec{x}_k}+1)}{2}}\corr{M^{(s_1)}_{x_{s_{1,1}}x_{s_{1,2}}\ldots x_{s_{1,N_1}}}\ldots M^{(s_k)}_{x_{s_{k,1}}x_{s_{k,2}}\ldots x_{s_{k,N_k}}}}\\
      =&\sum_{\substack{\vec{x}_{N-k}\in\DE{0,1}^{\times N-k}\\\omega_{\vec{x}_{N-k}} \text{ odd}}}(-1)^{\frac{\omega_{\vec{x}_{N-k}}\cdot(\omega_{\vec{x}_{N-k}}+1)}{2}} S_{k,\vec{x}_{N-k}}^- +\sum_{\substack{\vec{x}_{N-k}\in\DE{0,1}^{\times N-k}\\\omega_{\vec{x}_{N-k}} \text{ even}}}(-1)^{\frac{\omega_{\vec{x}_{N-k}}\cdot(\omega_{\vec{x}_{N-k}}+1)}{2}} S_{k,\vec{x}_{N-k}}^+
    \end{align}
    where the $S_{k,\vec{x}_{N-k}}^{\pm}$ is the k-partite Svetlichny inequality where the observables corresponding to each set $\mathcal{S}_i$ are labelled by $\vec{x}_{N-k}$ given in \eqref{eq:xn-k}.

    The rest of the proof follows as the proof of Proposition~\ref{prop:sk_violation}.
    \end{proof}
    {Similarly, the self-testing proof in Appendix \ref{app:net-selftest} can be generalised by replacing the operators $A_{x_i}^{(i)}$ with $M^{(s_i)}_{x_{s_{i,1}} x_{s_{i,2}}\ldots x_{s_{i,N_i}}}$ for $i = 1,\ldots,k-1$ and relabeling $M^{(\mathcal{D})}_{\vec{x}_{\mathcal{D}}}$ as $M^{(s_k)}_{x_{s_{k,1}} x_{s_{k,2}}\ldots x_{s_{k,N_k}}}$. This directly proves the following result
    \begin{theorem}\label{thm:self-test}
The maximum violation of the $N$-partite Svetlichny inequality where  the parties are grouped in $k$ disjoint collaborating subsets, $\mathcal{S}_1 \ldots \mathcal{S}_k$, self-tests that a $k-$partite GHZ state is shared by the different clusters. Moreover, the same correlations also self-tests a set of Pauli observables for the joint measurements performed by each cluster. 
\end{theorem}
}

\end{document}